\documentclass{article}
\usepackage[latin9]{inputenc}
\usepackage{float}
\usepackage{multirow}
\usepackage{amsmath}
\usepackage{amsthm}
\usepackage{amssymb}
\usepackage[]
 {hyperref}

\makeatletter

\providecommand{\tabularnewline}{\\}

\newcommand{\lyxaddress}[1]{
	\par {\raggedright #1
	\vspace{1.4em}
	\noindent\par}
}
\theoremstyle{plain}
\newtheorem{thm}{\protect\theoremname}
\theoremstyle{plain}
\newtheorem{prop}[thm]{\protect\propositionname}
\theoremstyle{remark}
\newtheorem{rem}[thm]{\protect\remarkname}
\theoremstyle{plain}
\newtheorem{cor}[thm]{\protect\corollaryname}

\makeatother

\providecommand{\corollaryname}{Corollary}
\providecommand{\propositionname}{Proposition}
\providecommand{\remarkname}{Remark}
\providecommand{\theoremname}{Theorem}

\begin{document}
\title{New Equivalence Tests for Approximate Independence in Contingency
Tables}
\author{Vladimir Ostrovski}
\maketitle

\lyxaddress{Ampega Asset Management GmbH, Charles-de-Gaulle-Platz 1, 50679 Cologne,
Germany\\
Email: vladimir.ostrovski77@googlemail.com}
\begin{abstract}
We introduce new equivalence tests for approximate independence in
two-way contingency tables. The critical values are calculated asymptotically.
The finite sample performance of the tests is improved by means of
the bootstrap. An estimator of boundary points is developed to make
the bootstrap based tests statistically efficient and computationally
feasible. We compare the performance of the proposed tests for different
table sizes by simulation. Then we apply the tests to real data sets.
The tests are implemented in R and available online, see \href{https://github.com/TestingEquivalence/EquivalenceTestIndependenceR}{https://github.com/TestingEquivalence/EquivalenceTestIndependenceR}.
\end{abstract}

\section{Introduction}

Testing for approximate row-column independence in two-way contingency
tables is a common task in statistical practice. The first publications
on this topic go back to \cite{HodgesLehmann1954}, \cite{DiaconisEfron1985}.
More recently, \cite{LiuLindsay2009} applied the semi-parametric
tubular tolerance regions to the row-column independence model in
two-way contingency tables. The method relies on the analytical properties
of the LRT statistic to obtain a closed form estimator of boundary
points. \cite{Wellek2010} develops a test for independence in multi-way
contingency tables in Section 9.2. For this purpose, he applies a
test for consistency with a fully specified multinomial distribution
as follows. First, the marginal distributions of the contingency table
are calculated. The test statistic is the Euclidean distance between
the product measure of the marginal distributions and the contingency
table. The critical value is calculated asymptotically.

\cite{Ostrovski2018} proposes a general method to test equivalence
to families of multinomial distributions, which is based on the minimum
distance 
\begin{equation}
d\left(p,\mathcal{M}\right)=\inf_{q\in\mathcal{M}}d\left(p,q\right)\label{min_dist}
\end{equation}
to a family $\mathcal{M}$ of multinomial distributions. If $d$ is
Euclidean distance and $\mathcal{M}$ is the independence model then
the calculation of minimum distance \eqref{min_dist} requires numerical
optimization. Generally, the method relies on the existence of a continuous
minimizer of \eqref{min_dist}. Unfortunately, it could not be shown
if a continuous minimizer exists for the independence model. Instead,
\cite{Ostrovski2018} assumes the existence of a continuous minimizer
at all points and then applies the method to test for approximate
independence. Additionally, numerical calculation of minimum distance
\eqref{min_dist} makes the bootstrap test computationally intensive.

We follow the lines of \cite{Ostrovski2018}, but avoid the numerical
valuation of the minimum distance \eqref{min_dist} in the special
case of independence testing. We also propose an efficient bootstrap
test, which is based on the randomized estimator of the boundary points.

Any two-way contingency table of the size $k_{1}\times k_{2}$ corresponds
to a probability matrix from $\mathbb{R}^{k_{1}}\times\mathbb{R}^{k_{2}}$.
Let $p=\left(p_{ij}\right)$ denote the probability matrix. Let $\mathcal{M}$
be the independence model, which contains all product measures of
the corresponding dimensions. The approximate row-column independence
can be shown by testing 
\begin{equation}
H_{0}=\left\{ d\left(p,\mathcal{M}\right)\geq\varepsilon\right\} \:\mbox{against}\:H_{1}=\left\{ d\left(p,\mathcal{M}\right)<\varepsilon\right\} ,\label{test_formal}
\end{equation}
where $\varepsilon>0$ is a tolerance parameter.

Let $r$ and $c$ denote the probability vectors of the marginal distributions,
which are defined by $r_{i}=\sum_{j=1}^{k_{2}}p_{ij}$ and $c_{j}=\sum_{i=1}^{k_{1}}p_{ij}$.
A probability matrix $p$ belongs to $\mathcal{M}$ iff the equality
$p_{ij}=r_{i}c_{j}$ is fulfilled for all $p_{ij}$. We consider the
transformations $h_{a}$ and $h_{r}$ of the matrix $p$, which are
defined by $h_{a}\left(p\right)=\left(p_{ij}-r_{i}c_{j}\right)$ and
$h_{r}\left(p\right)=\left(\frac{p_{ij}}{r_{i}c_{j}}\right).$

For any differentiable distance $l$ on $\mathbb{R}^{k_{1}}\times\mathbb{R}^{k_{2}}$
we define two new distances $d_{a}\left(p,q\right)=l\left(h_{a}\left(p\right),h_{a}\left(q\right)\right)$
and $d_{r}\left(p,q\right)=l\left(h_{r}\left(p\right),h_{r}\left(q\right)\right)$.
It should be noted that $d_{a}$ and $d_{r}$ are only pseudo-metrics
because $d_{r}\left(p,q\right)=0$ or $d_{a}\left(p,q\right)=0$ does
not imply $p=q$. We put these distances in Equation (\ref{min_dist})
and obtain 
\[
d_{a}\left(p,\mathcal{M}\right)=\inf_{q\in\mathcal{M}}d_{a}\left(p,q\right)=\inf_{q\in\mathcal{M}}l\left(h_{a}\left(p\right),0\right)=l\left(h_{a}\left(p\right),0\right)
\]
and 
\[
d_{r}\left(p,\mathcal{M}\right)=l\left(h_{r}\left(p\right),1\right),
\]
where $0$ denotes the zero matrix and $1$ is the matrix of ones.
The distances $d_{a}\left(p,\mathcal{M}\right)$ and $d_{r}\left(p,\mathcal{M}\right)$
can be interpreted respectively as the absolute deviation and the
relative deviation between $p$ and the product measure of the marginal
distributions. The distances $d_{a}\left(p,\mathcal{M}\right)$ and
$d_{r}\left(p,\mathcal{M}\right)$ are easy to calculate without optimization.

Therefore, $d_{a}$ and $d_{r}$ are good candidates for the general
distance $d$ in \eqref{min_dist} and we will use only these two
specific distances in remainder of the paper. 

We observe a contingency table $p_{n}$ of relative frequencies, where
$n$ is the sample size and $p$ is the true underlying probability
matrix. Then the test statistic for \eqref{test_formal} is 
\[
T_{a}\left(p_{n}\right)=\sqrt{n}\left(d_{a}\left(p_{n},\mathcal{M}\right)-\varepsilon\right)
\]
 or 
\[
T_{r}\left(p_{n}\right)=\sqrt{n}\left(d_{r}\left(p_{n},\mathcal{M}\right)-\varepsilon\right)
\]
depending on user preference. Below we write $d_{*}$ instead of $d_{a}$
and $d_{r}$ if the statements are correct for both distances. We
use the subscript $*$ instead of $a$ and $r$, if appropriate.

\section{Asymptotic Tests}

In this section, we derive the asymptotic distribution of the test
statistic and give a detailed description of the asymptotic test.

Let $v:\mathbb{R}^{k_{1}}\times\mathbb{R}^{k_{2}}\rightarrow\mathbb{R}^{k_{1}+k_{2}}$
be the usual bijection $v(p)=\left(p_{11},p_{12},\ldots,p_{k_{1}k_{2}}\right)$.
Let $\mathring{d}_{*}$ denote the derivative of the function $q\mapsto d_{*}\left(v^{-1}\left(q\right),\mathcal{M}\right)$,
which can be easily calculated using the chain rule. 
\begin{prop}
Let $p$ be a boundary point of $H_{0}$ and $q=v\left(p\right)$.
Let $D_{q}$ denote a square diagonal matrix, whose diagonal entries
are $q_{1},\ldots,q_{k_{1}+k_{2}}$. Then the asymptotic distribution
of $T_{*}\left(p_{n}\right)$ is Gaussian with mean zero and variance
$\sigma_{*}\left(p\right)=\mathring{d}_{*}\left(q\right)\Sigma\left(q\right)\mathring{d}_{*}\left(q\right)^{t}$,
where $\Sigma\left(q\right)=D_{q}-qq^{t}$ is a covariance matrix. 
\end{prop}

\begin{proof}
Let $q_{n}=v\left(p_{n}\right)$. The normalized vector $\sqrt{n}\left(q_{n}-q\right)$
converges weakly to a random variable, which is Gaussian with mean
zero and covariance matrix $\Sigma\left(q\right)$, see \cite{Bishop1975},
Theorem 14.3-4 for details. The assertion follows by the delta method,
see \cite{Vaart1998}, p. 26, Theorem 3.1. 
\end{proof}
The asymptotic variance $\sigma_{*}\left(p\right)$ is unknown and
can be estimated by $\sigma_{*}\left(p_{n}\right)$. The estimator
$\sigma_{*}\left(p_{n}\right)$ is consistent by the continuous mapping
theorem because $p\mapsto\sigma_{*}\left(p\right)$ is a continuous
function. Let $l_{\alpha}$ denote the lower $\alpha$-quantile of
the normal distribution. Then the critical value of the asymptotic
test is $l_{\alpha}\sigma_{*}\left(p_{n}\right)$. Now we have all
components of the asymptotic test, which can be carried out as follows: 
\begin{enumerate}
\item Given are the contingency table $p_{n}$ of relative frequencies,
the tolerance parameter $\varepsilon$ and the significance level
$\alpha$. 
\item Calculate the test statistic $T_{*}\left(p_{n}\right)=\sqrt{n}\left(d_{*}\left(p_{n},\mathcal{M}\right)-\varepsilon\right)$. 
\item Calculate the asymptotic variance $\sigma_{*}\left(p_{n}\right)$. 
\item Reject $H_{0}$ if $T_{*}\left(p_{n}\right)\leq l_{\alpha}\sigma_{*}\left(p_{n}\right)$. 
\end{enumerate}
The outlined test is locally asymptotically most powerful, see \cite{Ostrovski2016},
Proposition 3. 
\begin{rem}
The minimum tolerance parameter $\varepsilon$, for which the asymptotic
test rejects $H_{0}$, can be calculated as $d_{*}\left(p_{n},\mathcal{M}\right)-\frac{1}{\sqrt{n}}l_{\alpha}\sigma_{*}\left(p_{n}\right)$. 
\end{rem}

\begin{rem}
The asymptotic test can be straightforward generalized for the multi-way
contingency tables. 
\end{rem}

\section{Bootstrap Tests}\label{sec_bootstrap}

The parametric bootstrap is an efficient method to improve the finite
sample performance of the proposed tests. Let $\partial H_{0}$ denote
the boundary of $H_{0}$. Let $\tilde{p}_{n}$ denote an estimator
of $p$, which fulfills the condition $\tilde{p}_{n}\in\partial H_{0}$.
The critical value $c\left(\alpha,p\right)$ can be estimated by 
\[
c\left(\alpha,\tilde{p}_{n}\right)=\sup_{x\in\mathbb{R}}\left\{ x:P\left(T_{*}\left(\tilde{p}_{n}\right)\leq x\right)\leq\alpha\right\} 
\]
because the critical value should be estimated so as if $H_{0}$ were
true. The estimator $c\left(\alpha,\tilde{p}_{n}\right)$ can be computed
by the Monte Carlo method to any degree of accuracy.

The minimum distance estimator of $p$ would be difficult to compute
because the boundary $\partial H_{0}$ cannot be parameterized to
apply common optimization techniques. Therefore, we propose a computationally
feasible estimator of $p$, which is based on the randomized approximation
to the minimum distance estimator.

Let $q$ be some probability matrix such that $d_{*}\left(q,\mathcal{M}\right)>\varepsilon$.
If $d_{*}\left(p_{n},\mathcal{M}\right)\leq\varepsilon$, then let
$a_{n}$ be the largest number from $\left[0,1\right]$ such that
\[
d_{*}\left(a_{n}p_{n}+\left(1-a_{n}\right)q,\mathcal{M}\right)=\varepsilon.
\]
Otherwise let $a_{n}=1$. The linear combination 
\[
c\left(p_{n},q\right)=a_{n}p_{n}+\left(1-a_{n}\right)q
\]
 is a consistent estimator of the boundary point $p$ under additional
requirements as shown below.
\begin{prop}
Assume that $d_{*}\left(ap+\left(1-a\right)q,\mathcal{M}\right)>\varepsilon$
for all $a\in\left[0,1\right)$. Then $c\left(p_{n},q\right)\rightarrow p$
a.e. for $n\rightarrow\infty$.\label{consist_bnd_est} 
\end{prop}

\begin{proof}
We show that $a_{n}\rightarrow1$ for $n\rightarrow\infty$. Let $d_{*}\left(p_{n},\mathcal{M}\right)<\varepsilon$
because $a_{n}=1$ otherwise. The function $f:a\mapsto d_{*}\left(ap_{n}+\left(1-a\right)q,\mathcal{M}\right)$
is continuous on $\left[0,1\right]$ and $f\left(0\right)=d_{*}\left(q,\mathcal{M}\right)>\varepsilon$
as well as $f\left(1\right)=d_{*}\left(p_{n},\mathcal{M}\right)\leq\varepsilon$.
Therefore, there exists a largest number $a_{n}\in\left[0,1\right]$
such that $f\left(a_{n}\right)=\varepsilon$. It is worth mentioning
that $a_{n}$ is a function of $p_{n}$.

Let $E=\left\{ \lim_{n\rightarrow\infty}p_{n}=p\right\} $. By the
strong law of large numbers, $p_{n}$ converges to $p$ a.e. for $n\rightarrow\infty$
and therefore $P(E)=1$. Let $\omega\in E$ be an arbitrary point
and let $a_{n}\left(\omega\right)$ denote $a_{n}\left(f_{n}\left(\omega\right)\right)$.
The sequence of $a_{n}\left(\omega\right)$ is bounded. Hence, there
exists a convergent sub-sequence $a_{n_{j}}\left(\omega\right)\rightarrow a_{0}\left(\omega\right)$
for $j\rightarrow\infty$. We obtain 
\[
a_{n_{j}}\left(\omega\right)p_{n_{j}}\left(\omega\right)+\left(1-a_{n_{j}}\left(\omega\right)\right)q\rightarrow a_{0}\left(\omega\right)p+\left(1-a_{0}\left(\omega\right)\right)q
\]
for $j\rightarrow\infty$ and consequently 
\[
d_{*}\left(a_{0}\left(\omega\right)p+\left(1-a_{0}\left(\omega\right)\right)q,\mathcal{M}\right)=\varepsilon.
\]
We conclude $a_{0}\left(\omega\right)=1$ due to the assumption $d_{*}\left(ap+\left(1-a\right)q,\mathcal{M}\right)>\varepsilon$
for all $a\in\left[0,1\right)$. Overall, we have shown that $a_{n}\left(\omega\right)\rightarrow1$
for all $\omega\in E$. 
\end{proof}
Let $Q=\left\{ q_{1},\ldots,q_{m}\right\} $ be a finite set of probability
matrices, such that any $q_{i}$ fulfills $d_{*}\left(q_{i},\mathcal{M}\right)>\varepsilon$.
We define the estimator $\tilde{p}_{n}$ as a minimum distance estimator
among the linear combinations $c\left(p_{n},q\right)$ for all $q\in Q$.
Formally, the estimator $\tilde{p}_{n}$ equals $c\left(p_{n},q_{i}\right)$,
which fulfills the conditions $q_{i}\in Q$ and $l\left(c\left(p_{n},q_{i}\right),p_{n}\right)=\min_{q\in Q}l\left(c\left(p_{n},q\right),p_{n}\right)$.
Note that the distance $l$ is used to define the estimator $\tilde{p}_{n}$
because $d_{*}$ is a pseudo-metric only. 
\begin{cor}
Let at least one $q\in Q$ satisfy $d_{*}\left(ap+\left(1-a\right)q,\mathcal{M}\right)>\varepsilon$
for all $a\in\left[0,1\right)$. Then $\tilde{p}_{n}\rightarrow p$
a.e. for $n\rightarrow\infty$.
\end{cor}

\begin{proof}
By definition of $\tilde{p}_{n}$, we obtain 
\[
l\left(\tilde{p}_{n},p_{n}\right)\leq l\left(c\left(p_{n},q\right),p_{n}\right)\leq l\left(c\left(p_{n},q\right),p\right)+l\left(p,p_{n}\right),
\]
where $l\left(c\left(p_{n},q\right),p\right)\rightarrow0$ a.e. by
Proposition \ref{consist_bnd_est} and $l\left(p,p_{n}\right)\rightarrow0$
a.e. by the strong law of large numbers. 
\end{proof}
The bootstrap test can be carried out as follows: 
\begin{enumerate}
\item Given are the contingency table $p_{n}$ of relative frequencies,
the tolerance parameter $\varepsilon$, the number of exterior points
$m$ and the significance level $\alpha$. 
\item Calculate the test statistic $T_{*}\left(p_{n}\right)=\sqrt{n}\left(d_{*}\left(p_{n},\mathcal{M}\right)-\varepsilon\right)$. 
\item If $T_{*}\left(p_{n}\right)\geq0$ then set $\tilde{p}_{n}=p_{n}$
and go to step 7. 
\item Find $m$ different points $q\in H_{0}$ such that $d_{*}\left(q,\mathcal{M}\right)>\varepsilon$.
The following rejection algorithm can be applied for the search: 
\begin{enumerate}
\item Simulate a random matrix $w$ whose entries are independently uniformly
distributed on $\left[0,1\right]$. 
\item Normalize $w$ to a probability matrix $q$. 
\item Add $q$ to $Q$ if $d_{*}\left(q,\mathcal{M}\right)>\varepsilon$
or reject $q$ otherwise. 
\item Repeat previous steps until all exterior points are found. 
\end{enumerate}
\item Solve the equation $d_{*}\left(a_{n}p_{n}+(1-a_{n})q,\mathcal{M}\right)=\varepsilon$
for $a_{n}$ using some root finding method. Repeat for all $q\in Q$. 
\item Find the minimum distance estimator $\tilde{p}_{n}$ among all linear
combinations $c\left(p_{n},q\right)$, where $q\in Q$. 
\item Estimate the critical value $c\left(\alpha,\tilde{p}_{n}\right)$
using Monte Carlo simulation. 
\item Reject $H_{0}$ if $T_{*}\left(p_{n}\right)\leq c\left(\alpha,\tilde{p}_{n}\right)$. 
\end{enumerate}
\begin{rem}
The bootstrap test is asymptotically consistent, see \cite{Lehmann2008},
Theorem 15.6.1. Consequently, the test is also locally asymptotically
most powerful see \cite{Ostrovski2016}, Proposition 3. 
\end{rem}

\begin{rem}
The appropriate number of exterior points $m$ can be found empirically.
We found that $m=\left(k_{1}+k_{2}\right)*50$ is sufficient and scales
well with the table size. 
\end{rem}

\begin{rem}
The minimum tolerance parameter $\varepsilon$, for which the bootstrap
test rejects $H_{0}$, can be found numerically. For this purpose,
the equation $T_{*}\left(p_{n}\right)=c\left(\alpha,\tilde{p}_{n}\right)$
should be solved for the tolerance parameter $\varepsilon$ using
some root finding algorithm. The exterior points and bootstrap samples
should remain unchanged during optimization. 
\end{rem}

\section{Simulation Study of Finite Sample Performance}

We study the finite sample performance of the proposed tests by the
Monte Carlo simulation for different sample sizes and table sizes.
The tests are implemented in R and available online, see \href{https://github.com/TestingEquivalence/EquivalenceTestIndependenceR}{https://github.com/TestingEquivalence/EquivalenceTestIndependenceR}. 

The distance $l$ is scaled Euclidean distance $l_{2}$, where the
scale factor is necessary to obtain comparable test results for different
table sizes. We use $l=\frac{1}{\sqrt{k_{1}k_{2}}}l_{2}$ in case
of $d_{r}$ and $l=\sqrt{k_{1}k_{2}}l_{2}$ in case of $d_{a}$. Alternatively
the smoothed total variation distance would be a good choice, see
\cite{Ostrovski2016}.

The minimum $\varepsilon$, for which the test power equals $0.9$,
is calculated for different table sizes and sample sizes at the uniform
probability matrices for the purpose of throwing some light on the
appropriate values of $\varepsilon$ and the effective sample sizes.
Table \ref{tab_pow} shows the minimum $\varepsilon$ for the distance
$d_{r}$. The minimum $\varepsilon$ for $d_{a}$ can be found in
Table \ref{tab_pow_abs} because the results are very similar for
$d_{a}$ and $d_{r}$. The minimum $\varepsilon$ decreases with the
increasing sample size at the rate $n^{-\frac{1}{2}}$. The minimum
parameter $\varepsilon$ climbs with the increasing table size at
the rate $k_{1}+k_{2}$. Thus, the test power falls slowly with the
increasing table size. The bootstrap tests have a smaller minimum
$\varepsilon$ than the asymptotic tests and the difference increases
considerably with the table size.

We study the type I error rates at 100 randomly selected points from
$\partial H_{0}$ because the boundary of $H_{0}$ is a very complex
set and it is difficult to identify particularly interesting boundary
points. The points are found using steps 4 and 5 of the algorithm
at the end of Section \ref{sec_bootstrap}. The sample size $n$ equals
$100*\left(k_{1}+k_{2}\right)$ to maintain similar test power for
different table sizes because test power falls with increasing table
size. The simulation results are summarized in Table \ref{tab_bnd}.
The power of all tests varies considerably from point to point. The
averaged power of the asymptotic tests decreases quickly with the
table size. The asymptotic tests are not conservative for the small
tables and become very conservative for the larger tables. The averaged
power of the bootstrap tests is very close to the nominal level for
all table sizes. However, the bootstrap tests are not conservative
for all table sizes. Particularly, the $d_{r}$ based bootstrap test
shows strong anti-conservative tendency.

A detailed analysis of the boundary points shows that the test power
is far above the nominal level at the points, where $r_{i}c_{j}$
is close to zero for some $i$ and $j$. Therefore, the test results
should be treated with caution, if the marginal probability is close
to zero for at least one category.

The conservative tests can be obtained by shrinking the tolerance
parameter $\varepsilon$. Table \ref{tab_bnd_shrinkage} summarizes
the simulation results for $\varepsilon=0.18$, where the test power
is calculated at the same points as in Table \ref{tab_bnd}, i.e.
$d_{*}\left(p,\mathcal{M}\right)=0.2$ at all considered points. Then
the $d_{a}$ based tests are conservative at all points and the $d_{r}$
based tests are still non conservative at some points.

The type II error rates are studied at 100 randomly selected product
measures for each table size, see Table \ref{tab_prod}. It should
be noted that Table \ref{tab_prod} contains test power and the type
II error rate equals 1 minus test power. The sample size equals $100*\left(k_{1}+k_{2}\right)$
to be comparable to the type I error analysis. The power of $d_{r}$-based
tests changes very strongly from point to point. Given the fixed table
size, the power of the $d_{a}$-based tests is almost constant at
all considered points. The averaged power of the asymptotic tests
decreases slightly with the increasing table size. The averaged power
of the bootstrap tests does not change with the table size.

\begin{table}[H]
\centering %
\begin{tabular}{|c|ccccccc|}
\multicolumn{8}{c}{Asymptotic test}\tabularnewline
\hline 
table size  & 100  & 200  & 500  & 1000  & 2000  & 5000  & 10000\tabularnewline
\hline 
2$\times$4  & 0.390  & 0.272  & 0.171  & 0.120  & 0.085  & 0.054  & 0.038\tabularnewline
3$\times$3  & 0.418  & 0.295  & 0.185  & 0.130  & 0.092  & 0.058  & 0.041\tabularnewline
3$\times$4  & 0.474  & 0.331  & 0.208  & 0.146  & 0.104  & 0.066  & 0.046\tabularnewline
4$\times$4  & 0.540  & 0.375  & 0.236  & 0.166  & 0.117  & 0.074  & 0.052\tabularnewline
4$\times$5  & 0.593  & 0.413  & 0.258  & 0.181  & 0.128  & 0.081  & 0.057\tabularnewline
5$\times$5  & 0.655  & 0.453  & 0.283  & 0.200  & 0.141  & 0.089  & 0.063\tabularnewline
\hline 
\multicolumn{8}{c}{Bootstrap test}\tabularnewline
\hline 
2$\times$4  & 0.382  & 0.271  & 0.171  & 0.121  & 0.085  & 0.054  & 0.038\tabularnewline
3$\times$3  & 0.400  & 0.284  & 0.179  & 0.126  & 0.089  & 0.056  & 0.040\tabularnewline
3$\times$4  & 0.431  & 0.302  & 0.190  & 0.134  & 0.095  & 0.060  & 0.042\tabularnewline
4$\times$4  & 0.464  & 0.322  & 0.203  & 0.143  & 0.101  & 0.064  & 0.045\tabularnewline
4$\times$5  & 0.488  & 0.339  & 0.212  & 0.149  & 0.106  & 0.067  & 0.047\tabularnewline
5$\times$5  & 0.519  & 0.356  & 0.223  & 0.158  & 0.111  & 0.071  & 0.050\tabularnewline
\hline 
\end{tabular}\caption{Minimum tolerance parameter $\varepsilon$, for which the test power
equals $0.9$ at nominal level $\alpha=0.05$, is calculated at the
uniform probability matrices using the distance $d_{r}$. The sample
size is $100,\ldots,10000$.}\label{tab_pow}
\end{table}

\begin{table}[H]
\centering %
\begin{tabular}{|c|cccccc|cccccc|}
\multicolumn{1}{c}{} & \multicolumn{6}{c}{Asymptotic test based on $d_{r}$} & \multicolumn{6}{c}{Asymptotic test based on $d_{a}$}\tabularnewline
\hline 
 & 2$\times$4  & 3$\times$3  & 3$\times$4  & 4$\times$4  & 4$\times$5  & 5$\times$5  & 2$\times$4  & 3$\times$3  & 3$\times$4  & 4$\times$4  & 4$\times$5  & 5$\times$5\tabularnewline
\hline 
Minimum  & 0.02  & 0.01  & 0.01  & 0.00  & 0.00  & 0.00  & 0.04  & 0.03  & 0.02  & 0.01  & 0.00  & 0.00\tabularnewline
Maximum  & 0.12  & 0.10  & 0.06  & 0.05  & 0.03  & 0.01  & 0.08  & 0.08  & 0.06  & 0.03  & 0.02  & 0.01\tabularnewline
Average  & 0.05  & 0.04  & 0.02  & 0.01  & 0.01  & 0.00  & 0.06  & 0.05  & 0.03  & 0.02  & 0.01  & 0.00\tabularnewline
Deviation  & 0.02  & 0.01  & 0.01  & 0.01  & 0.00  & 0.00  & 0.01  & 0.01  & 0.01  & 0.00  & 0.00  & 0.00\tabularnewline
\hline 
\multicolumn{1}{c}{} & \multicolumn{6}{c}{Bootstrap test based on $d_{r}$} & \multicolumn{6}{c}{Bootstrap test based on $d_{a}$}\tabularnewline
\hline 
Minimum  & 0.00  & 0.00  & 0.01  & 0.03  & 0.03  & 0.04  & 0.04  & 0.03  & 0.04  & 0.04  & 0.04  & 0.04\tabularnewline
Maximum  & 0.07  & 0.08  & 0.10  & 0.13  & 0.10  & 0.11  & 0.09  & 0.08  & 0.07  & 0.07  & 0.07  & 0.06\tabularnewline
Average  & 0.04  & 0.05  & 0.05  & 0.05  & 0.05  & 0.06  & 0.05  & 0.05  & 0.05  & 0.05  & 0.05  & 0.05\tabularnewline
Deviation  & 0.01  & 0.01  & 0.01  & 0.01  & 0.01  & 0.01  & 0.01  & 0.01  & 0.01  & 0.00  & 0.00  & 0.00\tabularnewline
\hline 
\end{tabular}\caption{Summary of the simulated exact rejection probability of the equivalence
tests at nominal level $\alpha=0.05$ and tolerance parameter $\varepsilon=0.2$.
The rejection probability is simulated at 100 randomly selected boundary
points. The sample size is $\left(k_{1}+k_{2}\right)*100$ and the
number of replications is 10.000 for each experiment.}\label{tab_bnd}
\end{table}

\begin{table}[H]
\centering %
\begin{tabular}{|c|cccccc|cccccc|}
\multicolumn{1}{c}{} & \multicolumn{6}{c}{Asymptotic test based on $d_{r}$ } & \multicolumn{6}{c}{Asymptotic test based on $d_{a}$}\tabularnewline
\hline 
 & 2$\times$4  & 3$\times$3  & 3$\times$4  & 4$\times$4  & 4$\times$5  & 5$\times$5  & 2$\times$4  & 3$\times$3  & 3$\times$4  & 4$\times$4  & 4$\times$5  & 5$\times$5\tabularnewline
\hline 
Minimum  & 0.20  & 0.09  & 0.13  & 0.10  & 0.16  & 0.26  & 0.99  & 0.98  & 0.96  & 0.94  & 0.93  & 0.90\tabularnewline
Maximum  & 0.99  & 0.98  & 0.97  & 0.95  & 0.92  & 0.87  & 1.00  & 1.00  & 0.99  & 0.97  & 0.96  & 0.92\tabularnewline
Average  & 0.84  & 0.80  & 0.80  & 0.77  & 0.70  & 0.64  & 0.99  & 0.99  & 0.98  & 0.96  & 0.95  & 0.91\tabularnewline
Deviation  & 0.17  & 0.21  & 0.18  & 0.16  & 0.18  & 0.15  & 0.00  & 0.00  & 0.00  & 0.00  & 0.01  & 0.00\tabularnewline
\hline 
\multicolumn{1}{c}{} & \multicolumn{6}{c}{Bootstrap test based on $d_{r}$} & \multicolumn{6}{c}{Bootstrap test based on $d_{a}$}\tabularnewline
\hline 
Minimum  & 0.13  & 0.13  & 0.22  & 0.28  & 0.56  & 0.73  & 0.98  & 0.97  & 0.97  & 0.98  & 0.98  & 0.99\tabularnewline
Maximum  & 0.99  & 0.99  & 0.99  & 0.99  & 0.99  & 0.99  & 1.00  & 1.00  & 0.99  & 0.99  & 1.00  & 1.00\tabularnewline
Average  & 0.82  & 0.80  & 0.86  & 0.89  & 0.90  & 0.94  & 0.99  & 0.99  & 0.99  & 0.99  & 0.99  & 0.99\tabularnewline
Deviation  & 0.20  & 0.22  & 0.17  & 0.11  & 0.10  & 0.06  & 0.00  & 0.01  & 0.00  & 0.00  & 0.00  & 0.00\tabularnewline
\hline 
\end{tabular}\caption{Summary of the simulated test power at nominal level $\alpha=0.05$
and tolerance parameter $\varepsilon=0.2$. The rejection probability
is simulated at 100 randomly selected product measures. The sample
size is $\left(k_{1}+k_{2}\right)*100$ and the number of replications
is 10.000 for each experiment.}\label{tab_prod}
\end{table}

\section{Real data sets}

To demonstrate the application of the proposed tests, three examples
with real data sets are considered: gender and nitrendipine therapy
(Nitrendipine); eye color and hair color (Color); children number
and income (Children). The corresponding two way contingency tables
are given in Appendix, Tables \ref{tab:nitrendipine}, \ref{tab:color}
and \ref{tab:children}. Table \ref{tab_rw} displays the minimum
tolerance parameter $\varepsilon$, for which $H_{0}$ can be rejected
at the nominal level $\alpha=0.05$. The three examples are also used
in \cite{Ostrovski2018}, such that a direct comparison is possible.
The results for distance $d_{a}$ are similar to those presented in
\cite{Ostrovski2018} after appropriate re-scaling. However, we avoid
the unproven assumptions and the extensive use of the numeric optimization,
which are necessary in \cite{Ostrovski2018}.

The first example concerns with the question if the treatment outcome
on nitrendipine mono-therapy in patients suffering from mild arterial
hypertension depends on gender. The data set is also an example for
approximate independence in \cite{Wellek2010}. The asymptotic and
bootstrap test results for $d_{r}$ are very close to each other.
The results for $d_{a}$ differ considerably for the asymptotic and
bootstrap test. Given the small sample size, the treatment outcome
and gender can be considered approximately independent.

A common example for independence testing is the cross-classification
of eye color and hair color, see \cite{DiaconisEfron1985}, \cite{LiuLindsay2009}.
The test results in Table \ref{tab_rw} reflect the well known fact
that eye color and hair color are not independently distributed. All
tests behave very similarly and can reject $H_{0}$ only for very
large values of $\varepsilon$.

The cross-classification of the number of children by the annual income
has a large sample size. However, the category, where the number of
children is larger than or equal $4$, is sparsely populated. Therefore,
the $d_{r}$ based tests can reject $H_{0}$ only for comparatively
large values of $\varepsilon$ and the test power is low. The $d_{a}$
based tests show that the number of children and annual income may
be considered approximately independent, but the approximation is
very inaccurate.

\begin{table}[H]
\centering %
\begin{tabular}{|c|c|cc|cc|}
\hline 
\multirow{1}{*}{Data set} & \multirow{1}{*}{$n$} & \multicolumn{2}{c|}{Tests based on $d_{r}$} & \multicolumn{2}{c|}{Tests based on $d_{a}$}\tabularnewline
 &  & Asymptotic  & Bootstrap  & Asymptotic  & Bootstrap\tabularnewline
\hline 
Nitrendipine  & 217  & 0.228  & 0.214  & 0.147  & 0.141 \tabularnewline
Color  & 592  & 0.588  & 0.574  & 0.593  & 0.583 \tabularnewline
Children  & 25263  & 0.277  & 0.281  & 0.178  & 0.177\tabularnewline
\hline 
\end{tabular}

\caption{Minimum tolerance parameter $\varepsilon$, for which $H_{0}$ can
be rejected at the nominal level $\alpha=0.05$.}\label{tab_rw}
\end{table}

\begin{table}[H]
\centering %
\begin{tabular}{|c|ccccccc|}
\multicolumn{8}{c}{asymptotic test}\tabularnewline
\hline 
table size  & 100  & 200  & 500  & 1000  & 2000  & 5000  & 10000\tabularnewline
\hline 
2$\times$4  & 0.377  & 0.267  & 0.169  & 0.120  & 0.085  & 0.054  & 0.038\tabularnewline
3$\times$3  & 0.405  & 0.290  & 0.183  & 0.130  & 0.092  & 0.058  & 0.041\tabularnewline
3$\times$4  & 0.458  & 0.327  & 0.207  & 0.146  & 0.104  & 0.065  & 0.046\tabularnewline
4$\times$4  & 0.520  & 0.368  & 0.234  & 0.165  & 0.117  & 0.074  & 0.052\tabularnewline
4$\times$5  & 0.567  & 0.403  & 0.255  & 0.180  & 0.128  & 0.081  & 0.057\tabularnewline
5$\times$5  & 0.623  & 0.444  & 0.281  & 0.199  & 0.141  & 0.089  & 0.063\tabularnewline
\hline 
\multicolumn{8}{c}{bootstrap test}\tabularnewline
\hline 
2$\times$4  & 0.378  & 0.268  & 0.170  & 0.120  & 0.085  & 0.054  & 0.038\tabularnewline
3$\times$3  & 0.396  & 0.282  & 0.178  & 0.125  & 0.089  & 0.056  & 0.040\tabularnewline
3$\times$4  & 0.423  & 0.302  & 0.190  & 0.134  & 0.095  & 0.060  & 0.042\tabularnewline
4$\times$4  & 0.457  & 0.320  & 0.203  & 0.143  & 0.101  & 0.064  & 0.045\tabularnewline
4$\times$5  & 0.478  & 0.337  & 0.212  & 0.149  & 0.106  & 0.067  & 0.047\tabularnewline
5$\times$5  & 0.501  & 0.354  & 0.222  & 0.158  & 0.111  & 0.070  & 0.050\tabularnewline
\hline 
\end{tabular}\caption{Minimum tolerance parameter $\varepsilon$, for which the test power
equals $0.9$ at nominal level $\alpha=0.05$, is calculated at the
uniform probability matrices using the distance $d_{a}$. The sample
size is $100,\ldots,10000$.}\label{tab_pow_abs}
\end{table}

\begin{table}[H]
\centering %
\begin{tabular}{|c|cccc|}
\hline 
\multirow{1}{*}{gender} & \multicolumn{4}{c|}{outcome category}\tabularnewline
\hline 
 & 1  & 2  & 3  & 4\tabularnewline
\hline 
female  & 9  & 13  & 13  & 48\tabularnewline
male  & 24  & 18  & 20  & 72\tabularnewline
\hline 
\end{tabular}

\caption{Contingency table relating gender and treatment outcome on nitrendipine
mono-therapy in patients suffering from mild arterial hypertension. }\label{tab:nitrendipine}
\end{table}

\begin{table}[H]
\centering %
\begin{tabular}{|c|cccc|}
\hline 
\multirow{1}{*}{Eye color} & \multicolumn{4}{c|}{Hair color}\tabularnewline
\hline 
 & Black  & Brunette  & Red  & Blonde\tabularnewline
\hline 
Brown  & 68  & 119  & 26  & 7\tabularnewline
Blue  & 20  & 84  & 17  & 94\tabularnewline
Hazel  & 15  & 54  & 14  & 10\tabularnewline
Green  & 5  & 29  & 14  & 16\tabularnewline
\hline 
\end{tabular}

\caption{Cross-classification of eye color and hair color.}\label{tab:color}
\end{table}

\begin{table}[H]
\centering %
\begin{tabular}{|c|cccc|}
\hline 
\multirow{1}{*}{No. of children} & \multicolumn{4}{c|}{Annual income}\tabularnewline
\hline 
 & 0--1  & 1--2  & 2--3  & 3+\tabularnewline
\hline 
0  & 2161  & 3577  & 2184  & 1636\tabularnewline
1  & 2755  & 5081  & 2222  & 1052\tabularnewline
2  & 936  & 1753  & 640  & 306\tabularnewline
3  & 225  & 419  & 96  & 38\tabularnewline
4+  & 39  & 98  & 31  & 14\tabularnewline
\hline 
\end{tabular}\caption{Cross-classification of number of children by annual income.}\label{tab:children}
\end{table}

\begin{table}[H]
\centering %
\begin{tabular}{|c|cccccc|cccccc|}
\multicolumn{1}{c}{} & \multicolumn{6}{c}{Asymptotic test based on $d_{r}$} & \multicolumn{6}{c}{Asymptotic test based on $d_{a}$}\tabularnewline
\hline 
 & 2$\times$4  & 3$\times$3  & 3$\times$4  & 4$\times$4  & 4$\times$5  & 5$\times$5  & 2$\times$4  & 3$\times$3  & 3$\times$4  & 4$\times$4  & 4$\times$5  & 5$\times$5\tabularnewline
\hline 
Minimum  & 0.01  & 0.00  & 0.00  & 0.00  & 0.00  & 0.00  & 0.01  & 0.01  & 0.00  & 0.00  & 0.00  & 0.00\tabularnewline
Maximum  & 0.08  & 0.06  & 0.03  & 0.02  & 0.01  & 0.00  & 0.04  & 0.04  & 0.02  & 0.01  & 0.00  & 0.00\tabularnewline
Average  & 0.02  & 0.02  & 0.01  & 0.00  & 0.00  & 0.00  & 0.02  & 0.02  & 0.01  & 0.00  & 0.00  & 0.00\tabularnewline
Deviation  & 0.01  & 0.01  & 0.00  & 0.00  & 0.00  & 0.00  & 0.01  & 0.01  & 0.00  & 0.00  & 0.00  & 0.00\tabularnewline
\hline 
\multicolumn{1}{c}{} & \multicolumn{6}{c}{Bootstrap test based on $d_{r}$} & \multicolumn{6}{c}{Bootstrap test based on $d_{a}$}\tabularnewline
\hline 
Minimum  & 0.00  & 0.00  & 0.00  & 0.01  & 0.01  & 0.01  & 0.01  & 0.01  & 0.01  & 0.01  & 0.01  & 0.01\tabularnewline
Maximum  & 0.05  & 0.06  & 0.09  & 0.12  & 0.06  & 0.07  & 0.04  & 0.04  & 0.03  & 0.03  & 0.02  & 0.02\tabularnewline
Average  & 0.02  & 0.02  & 0.02  & 0.02  & 0.02  & 0.02  & 0.02  & 0.02  & 0.02  & 0.02  & 0.02  & 0.01\tabularnewline
Deviation  & 0.01  & 0.01  & 0.01  & 0.01  & 0.01  & 0.01  & 0.01  & 0.01  & 0.00  & 0.00  & 0.00  & 0.00\tabularnewline
\hline 
\end{tabular}\caption{Summary of the simulated exact rejection probability of the equivalence
tests at nominal level $\alpha=0.05$ and shrunk tolerance parameter
$\varepsilon=0.18$. The rejection probability is simulated at 100
randomly selected boundary points of $H_{0}=\left\{ d_{*}\left(p,\mathcal{M}\right)\protect\geq0.2\right\} $.
The sample size is $\left(k_{1}+k_{2}\right)*100$ and the number
of replications is 10.000 for each experiment.}\label{tab_bnd_shrinkage}
\end{table}

\bibliographystyle{plainnat}
\bibliography{VO_literature}

\end{document}